\documentclass[3p]{elsarticle}
\usepackage[T1]{fontenc}
\usepackage{graphicx, algorithm, algorithmic}
\usepackage[utf8]{inputenc}
\usepackage{tikz}
\usetikzlibrary{shapes.geometric,shapes.misc}
\usetikzlibrary{calc,shapes,arrows,arrows.meta,shapes.misc,decorations.pathreplacing}
\usetikzlibrary{positioning,calc}
\usepackage{tikz-qtree}
\usepackage{subcaption}
\usepackage{lipsum}
\usepackage{listings}
\usepackage{adjustbox}
\usepackage{amsmath,amssymb}
\usepackage{hyperref}

\newtheorem{proposition}{Proposition}

\newtheorem{theorem}{Theorem}
\newtheorem{lemma}{Lemma}
\newtheorem{corollary}{Corollary}

\newproof{proof}{Proof}

\def\Left{\ensuremath{\mathrm{Left}}}
\def\Right{\ensuremath{\mathrm{Right}}}
\def\First{\ensuremath{\mathrm{First}}}
\def\balpha{{\boldsymbol{\alpha}}}
\def\bbeta{{\boldsymbol{\beta}}}

\let\doendproof\endproof
\renewcommand\endproof{~\hfill~$\Box$\doendproof}

\usepackage{xcolor}

\definecolor{dkgreen}{rgb}{0,0.6,0}
\definecolor{gray}{rgb}{0.5,0.5,0.5}
\definecolor{mauve}{rgb}{0.58,0,0.82}

\lstdefinestyle{myScalastyle}{
	frame=tb,
	language=scala,
	aboveskip=3mm,
	belowskip=3mm,
	showstringspaces=false,
	columns=flexible,
	basicstyle={\small\ttfamily},
	numbers=none,
	numberstyle=\tiny\color{gray},
	keywordstyle=\color{blue},
	commentstyle=\color{dkgreen},
	stringstyle=\color{mauve},
	frame=single,
	breaklines=true,
	breakatwhitespace=true,
	tabsize=3,
	numbers=left,
	deletekeywords={for},
	otherkeywords={function,==,=, Let}
}

\begin{document}
\title{The ordinal generated by an ordinal grammar is computable}
		\author[kitti]{Kitti~Gelle}
		\author[szabivan]{Szabolcs~Iv\'an}
		\address[kitti,szabivan]{University of Szeged, Hungary}
		\tnotetext[thanksfn]{
			Ministry of Human Capacities, Hungary grant 20391-3/2018/FEKUSTRAT is acknowledged.
			Szabolcs Iv\'an was supported by the J\'anos Bolyai Scholarship of the Hungarian Academy of Sciences.} 
		
\begin{abstract}
A prefix grammar is a context-free grammar whose nonterminals generate prefix-free languages.
A prefix grammar $G$ is an ordinal grammar if the language $L(G)$ is well-ordered with respect to the lexicographic ordering.
It is known that from a finite system of parametric fixed point equations over ordinals one can construct an ordinal grammar $G$
such that the lexicographic order of $G$ is isomorphic with the least solution of the system, if this solution is well-ordered.
In this paper we show that given an ordinal grammar, one can compute (the Cantor normal form of)
the order type of the lexicographic order of its language, yielding that least solutions of fixed point equation systems
defining algebraic ordinals are effectively computable (and thus, their isomorphism problem is also decidable).
\end{abstract}
		
\maketitle
\section{Introduction}

Least solutions of finite systems of fixed points equations occur frequently in computer science.
Some very well-known instances of this are the regular and context-free languages, rational and algebraic power
series, well-founded semantics of generalized logic programs, semantics of functional programs, just to name
a few. A perhaps less-known instance is the notion of the algebraic linear orders of~\cite{Bloom:2010:MTC:1655414.1655555}.
A linear ordering is algebraic if it is (isomorphic to) the first component of the least solution of a finite system
of fixed point equations of the sort
\[F_i(x_0,\ldots,x_{n_i-1})=t_i,\quad i=1,\ldots,n,\]
where $n_1=0$ and each $t_i$ is an expression composed of the function variables $F_j$, $j=1,\ldots,n$,
the variables $x_0,\ldots,x_{n_i-1}$ which range over linear orders, the constant $1$ and the sum operation $+$.
As an example, consider the following system from~\cite{DBLP:journals/fuin/BloomE10}:
\begin{align*}
F_0 &= G(1)\\
G(x)&=x+G(F(x))\\
F(x)&=x+F(x)
\end{align*}
In this system, the function $F$ maps a linear order $x$ to $x+x+\ldots = x\times\omega$,
the function $G$ maps a linear order $x$ to $x+G(x\times \omega)=x+x\times\omega+G(x\times\omega^2)+\ldots
=x\times\omega^\omega$, thus the first component of the least solution of the system is $F_0=G(1)=\omega^\omega$.

If the system in question is parameterless, that is, $n_i=0$ for each $i$, then the ordering which it defines 
is called a regular ordering. An ordinal is called algebraic (regular, respectively) if it is algebraic (regular, resp.)
as a linear order. It is known~\cite{ITA_1980__14_2_131_0,BLOOM2001533,10.1007/978-3-540-73859-6_1,DBLP:journals/fuin/BloomE10,10.1007/978-3-642-29344-3_25}
that an ordinal is regular if and only if it is smaller than $\omega^\omega$ and is algebraic if and only if
it is smaller than $\omega^{\omega^\omega}$.

To prove the latter statement, the authors of~\cite{DBLP:journals/fuin/BloomE10} applied a path first used
by Courcelle~\cite{ITA_1978__12_4_319_0}: every countable linear order is isomorphic to the frontier
of some (possibly) infinite (say, binary) tree. Frontiers of infinite binary trees in turn
correspond to prefix-free languages over the binary alphabet, equipped with the lexicographic ordering.
Moreover, algebraic (regular, resp.) ordinals are exactly the lexicographic orderings of
context-free (regular, resp.) prefix-free languages~\cite{COURCELLE198395}
(prefix-free being optional here as each language
can be effectively transformed to a prefix-free order-isomorphic one for both the regular and the
algebraic case). Thus, studying lexicographic orderings of prefix-free regular or context-free languages
can give insight to regular or algebraic linear orders. The works~\cite{BLOOM2001533,6a0baa0d4e6744d38956d22057b410ce,BLOOM200555,10.1007/978-3-540-73859-6_1,
COURCELLE198395,ITA_1980__14_2_131_0,LOHREY201371,ITA_1986__20_4_371_0} deal with regular linear orders this way, in particular~\cite{LOHREY201371}
shows that the isomorphism problem for regular linear orders is decidable in polynomial time
The study of the context-free case was initiated in~\cite{10.1007/978-3-540-73859-6_1}, and further developed
in~\cite{DBLP:journals/fuin/BloomE10, doi:10.1142/S0129054111008155,ESIK2011107,10.1007/978-3-642-22321-1_19,10.1007/978-3-642-29344-3_25,CARAYOL2013285,KUSKE201446}.

Highlighting the results from these works that are tightly connected to the current paper: the case of
regular linear orders is well-understood, even their isomorphism problem (that is, whether
two regular linear orders, given by two finite sets of fixed-point equations, are isomorphic) 
is decidable. For algebraic linear orders, there are negative results: it is already undecidable whether
an algebraic linear ordering is dense, thus (as there are exactly four dense countable linear orders
up to isomorphism) the isomorphism problem of algebraic linear orders is undecidable.
On the other hand, deciding whether an algebraic linear order is scattered, or a well-order, is decidable.
The frontier of decidability of the isomorphism problem of algebraic linear orderings is an interesting
question: for the general case it is undecidable, while for the case of regular ordinals it is
known to be decidable by~\cite{LOHREY201371} and~\cite{Khoussainov:2005:ALO:1094622.1094625}.
In~\cite{DBLP:journals/fuin/BloomE10}, it was shown that a system of equations defining an algebraic ordering
can be effectively transformed (in polynomial time) to a so-called prefix grammar $G$
(a context-free grammar whose nonterminals each generate a prefix-free language), 
such that the lexicographic order of the language generated by $G$ is isomorphic to the algebraic ordering in question.
If the ordering is a well-ordering (i.e. the system defines an algebraic ordinal),
then the grammar we get is called an ordinal grammar, that is, a prefix
grammar generating a well-ordered language with respect to the lexicographic ordering.

In this paper we show that given an ordinal grammar, the order type of the lexicographic ordering of the
language it generates is computable (that is, we can effectively construct its Cantor normal form).
Hence, applying the above transformation we get that the Cantor normal form of any algebraic ordinal
is computable from its fixed-point system presentation, thus in particular, the isomorphism problem
of algebraic ordinals is decidable.

\section{Notation}	
When $n\geq 0$ is an integer, $[n]$ denotes the set $\{1,\ldots,n\}$. (Thus, $[0]$ is another notation for the empty set $\emptyset$.)
\subsection*{Linear orders, ordinals}

In this paper we consider countable linear orderings. A good reference on the topic is~\cite{rosenstein}.
A linear ordering $(I,<)$ is a set $I$ equipped with a strict linear order: an irreflexive, transitive and trichotome relation $<$.
When the order $<$ is clear from the context, we omit it. Set-theoretic properties of $I$ are lifted to $(I,<)$, thus we can say that
a linear order is finite, countable etc.
When $(I_1,<_1)$ and $(I_2,<_2)$ are linear orders, their (ordered) sum is $(I_1,<_1)+(I_2,<_2)=(I_1\uplus I_2,<)$ with
$x<y$ if and only if either $x\in I_1$ and $y\in I_2$, or $x,y\in I_1$ and $x<_1y$, or $x,y\in I_2$ and $x<_2y$.
A linear ordering $(I',<')$ is a subordering of $(I,<)$ if $I'\subseteq I$ and $<'$ is the restriction of $<$ onto $I'$.
In order to ease notation, we usually use $<$ in these cases in place of $<'$ and so we will simply write $(I_1,<)+(I_2,<)=(I,<)$
or even $I_1+I_2=I$ in the case of sums.

A linear ordering $I$ is called a \emph{well-ordering} if there are no infinite descending chains $\ldots<x_2<x_1<x_0$ in $I$. 
Clearly, well-orderings are closed under (finite) sums and suborderings,
and they are also closed under $\omega$-sums: if $I_1,I_2,\ldots$ are pairwise disjoint linear orderings, then their sum
$I=I_1+I_2+\ldots$ is the ordering with underlying set $\bigcup_i I_i$ and order $x<y$ if and only if $x\in I_i$ and $y\in I_j$
for some $i<j$, or $x,y\in I_i$ for some $i$ and $x<_iy$, which is well-ordered if so is each $I_i$.

Two linear orders $(I,<_i)$ and $(J,<_j)$ are called \emph{isomorphic} if there is a bijection $h:I\to J$ with $x<_iy$ implying
$h(x)<_jh(y)$. 
An \emph{order type} is an isomorphism class of linear orderings. The order type of the linear order $I$ is denoted by $o(I)$.
Clearly, if two orderings are isomorphic and one of them is a well-ordering, then so is the other one.
The \emph{ordinals} are the order types of well-orderings (for a concise introduction see e.g. the lecture notes of
J.~A.~Stark~\cite{jalex}). The order types of the finite ordered sets are identified with the
nonnegative integers. The order type of the natural numbers themselves (whose set is $\mathbb{N}_0=\{0,1,\ldots\}$,
equipped by their usual ordering)
is denoted by $\omega$, while the order types of the integers and rational numbers are respectively denoted by $\zeta$ and $\eta$.
Since if $o(I)=o(I')$ and $o(J)=o(J')$, then $o(I+J)=o(I'+J')$, the sum operation can be lifted to order types,
even for $\omega$-sums. For example, $\omega+\omega$ is the order type of $\{0,1\}\times\mathbb{N}$, equipped with the 
lexicographic ordering $(b_1,n_1)<(b_2,n_2)$ if and only if either $b_1<b_2$ or ($b_1=b_2$ and $n_1<n_2$). Note that $1+\omega=\omega$
but $\omega+1\neq\omega$.

The ordinals themselves are also equipped with a relation $<$ so that each set of ordinals is well-ordered by $<$, namely
$o_1<o_2$ if $o_1\neq o_2$ and there are linear orderings $I$ and $J$ such that $o(I)=o_1$, $o(J)=o_2$ and $I$ is a subordering of $J$.
With respect to this relation, every set $\Omega$ of ordinals have a least upper bound (a \emph{supremum}) $\bigvee\Omega$ (which is
also an ordinal), moreover, for each ordinal $\alpha$, the ordinals smaller than $\alpha$ form a set.

Each ordinal $\alpha$ is either a \emph{successor ordinal} in which case $\alpha=\beta+1$ for some smaller ordinal $\beta$,
or a \emph{limit ordinal} in which case $\alpha=\mathop\bigvee\limits_{\beta<\alpha}\beta$, the supremum of all the ordinals
smaller than $\alpha$. These two cases are disjoint. For an example, $0=\bigvee\emptyset$ is a limit ordinal, and it is the
smallest ordinal; $1$, $2$ and $42$ are successor ordinals, $\omega$ is a limit ordinal, $\omega+1$ is again a successor ordinal,
$\omega+\omega$ is a limit ordinal and so on.

Since every set of ordinals is well-ordered, and to each ordinal $\alpha$ the ordinals smaller than $\alpha$ form a set,
the principle of \emph{(well-founded) induction} is valid for ordinals: if $P$ is a property of ordinals, and
\begin{itemize}
	\item whenever $P$ holds for $\alpha$, then $P$ holds for $\alpha+1$ and
	\item whenever $\alpha$ is a limit ordinal and $P$ holds for each ordinal $\beta<\alpha$, then $P$ holds for $\beta$,
\end{itemize}
then $P$ holds for all the ordinals. (In practice we usually separate the case of $\alpha=0$ from the rest of the limit
ordinals.)

Over ordinals, the operations of (binary) product and exponentiation are defined via induction as follows:
\begin{align*}
\alpha\times 0 &=0 & \alpha\times(\beta+1)&=\alpha\times\beta+\alpha& \alpha\times\beta^*&=\mathop\bigvee\limits_{\beta'<\beta^*}\left(\alpha\times\beta'\right)\\
\alpha^0 &=1 & \alpha^{\beta+1}&=\alpha^\beta\times\alpha&\alpha^{\beta^*}&=\mathop\bigvee\limits_{\beta'<\beta^*}\alpha^{\beta'}
\end{align*}
where the equations of the last column hold for limit ordinals $\beta^*$.

Every ordinal $\alpha$ can be uniquely written as a finite sum
\[\alpha=\omega^{\alpha_1}\times n_1+\omega^{\alpha_2}\times n_2+\ldots+\omega^{\alpha_k}\times n_k\]
where $k\geq 0$ and for each $1\leq i\leq k$, $n_i>0$ are integers, and $\alpha_1>\alpha_2>\ldots>\alpha_k$ are ordinals.
The ordinal $\alpha_1$ in this form is called the \emph{degree} of $\alpha$, denoted by $\deg(\alpha)$, and the sum itself is
called the \emph{Cantor normal form} of $\alpha$.  The operations $+$ and $\times$ are associative, and the above operations satisfy the identities
\begin{align*}
\alpha\times(\beta+\gamma)&=\alpha\times\beta+\alpha\times\gamma&\alpha^\beta\times\alpha^\gamma&=\alpha^{\beta+\gamma}&(\alpha^\beta)^\gamma&=\alpha^{\beta\times\gamma}\\
\deg(\alpha+\beta)&=\max\{\deg(\alpha),\deg(\beta)\}&\deg(\alpha\times\beta)&=\deg(\alpha)+\deg(\beta)&\deg(\alpha^\beta)&=\deg(\alpha)\times\beta,
\end{align*}
the last one being valid only when $\alpha\geq\omega$. From $\deg(\alpha+\beta)=\max\{\deg(\alpha),\deg(\beta)\}$ we get that if
$o_1\leq o_2\leq \ldots$ are ordinals with $\deg(o_i)<\alpha$ for some ordinal $\alpha$, then $\deg(o_1+o_2+\ldots)\leq \alpha$
and equality holds if and only if $\bigvee\deg(o_i)=\alpha$ is a limit ordinal, in which case $o_1+o_2+\ldots=\omega^{\alpha}$.

The following theorem from~\cite{doi:10.1112/plms/s3-4.1.177} gives lower and upper bounds for the order type of the union of two well-ordered sets:
\begin{theorem}
	\label{thm-union}
	Let $(I,<)$ be a countable well-ordered set and $I=A\cup B$. Let us write the order types of $A$ and $B$ as
	\begin{align*}
	o(A) &= \omega^{\alpha_1}\times a_1+\ldots \omega^{\alpha_n}\times a_n\\
	o(B) &= \omega^{\alpha_1}\times b_1+\ldots \omega^{\alpha_n}\times b_n
	\end{align*}
	for an integer $n\geq 0$, ordinals $\alpha_1>\alpha_2>\ldots> \alpha_n$ and integer coefficients
	$a_1,\ldots,a_n,b_1,\ldots,b_n\geq 0$ such that $\max\{a_i,b_i\}\geq 1$ for each $1\leq i\leq n$.
	
	Then
	\begin{align*}
	o(I) &= \omega^{\alpha_1}\times c_1+\omega^{\alpha_2}\times c_2+\ldots +\omega^{\alpha_n}\times c_n
	\end{align*}
	for some integer coefficients $0\leq c_1,\ldots,c_n$ with $c_i\leq a_i+b_i$ for each $1\leq i\leq n$,
	and $c_1\geq\max\{a_1,b_1\}$.
\end{theorem}
Observe that the Theorem can be applied as follows: if $o(A)<\omega^\alpha\times N$ and $o(B)<\omega^\beta\times M$,
then $o(A\cup B)< \omega^{\max\{\alpha,\beta\}}\times(N+M-1)$: writing out the Cantor normal forms explicitly for
$o(A)$ and $o(B)$ we would get the coefficients for $\omega^{\max\{\alpha,\beta\}}$ can be at most $N-1$ and $M-1$,
respectively, making its coefficient in $o(I)$ to be at most $M+N-2$, thus (as the main term cannot be larger than $\omega^{\max\{\alpha,\beta\}}$ in either one of $o(A)$ and $o(B)$) we get $o(I)<\omega^{\max\{\alpha,\beta\}}\times(M+N-1)$. In particular, $\deg(o(A\cup B))=\max\{\deg(o(A)),\deg(o(B))\}$.

\subsection*{Order types of context-free languages}
For a nonempty finite set (an \emph{alphabet}) $\Sigma$ of terminal symbols, also called \emph{letters} equipped with a total ordering $<$,
let $\Sigma^*$ denote the set of all finite words $a_1a_2\ldots a_n$,
with $\varepsilon$ standing for the case $n=0$, the \emph{empty word}, and let $\Sigma^\omega$ denote the set of all
\emph{$\omega$-word}s $a_1a_2\ldots$. The set of all finite and $\omega$-words is $\Sigma^{\leq\omega}=\Sigma^\omega\cup\Sigma^*$.
When $u=a_1\ldots a_n$ is a finite word and $v=b_1b_2\ldots$ is either a finite or an $\omega$-word, then their product
is the word $u\cdot v=a_1\ldots a_nb_1b_2\ldots$, also written $uv$. Also, when $u=a_1\ldots a_n$ is a finite word, then its
\emph{$\omega$-power} is the word $u^\omega=a_1\ldots a_na_1\ldots a_na_1\ldots$ which is $\varepsilon$ if $u=\varepsilon$
and is an $\omega$-word whenever $u$ is nonempty.

Two (strict) partial orderings, the \emph{strict ordering} $<_s$ and the \emph{prefix ordering} $<_p$ are defined over $\Sigma^{\leq\omega}$
as follows:
\begin{itemize}
	\item  $u<_sv$ if and only if $u=u_1au_2$ and $v=u_1bv_2$ for some words $u_1\in\Sigma^*$, $u_2,v_2\in\Sigma^{\leq\omega}$
	and terminal symbols $a<b$
	\item $u<_pv$ if and only if $v=uw$ for some nonempty word $w\in\Sigma^{\leq\omega}$ (in particular, this implies $u\in\Sigma^*$).
\end{itemize}
The union of these partial orderings, the \emph{lexicographical ordering} $<_\ell~=~<_s\cup<_p$, simply written as $<$ when
it is clear from the context,
is a total ordering on $\Sigma^{\leq\omega}$, which is a complete lattice with respect to $<_\ell$.

A \emph{language} is an arbitrary set $L\subseteq\Sigma^*$ of \emph{finite} words. The \emph{supremum} of $L$, viewed as a subset
of $\bigl(\Sigma^{\leq \omega},<_\ell\bigr)$ is denoted by $\bigvee L$ and is either a finite word $u\in L$, or an $\omega$-word.
The \emph{order type} $o(L)$ of $L$ is the order type of the linear ordering $(L,\leq_\ell)$. 
As an example, the order types of the languages $a^*$, $a^*\cup\{b\}$ and
$b^*a^*$ are $\omega$, $\omega + 1$ and $\omega^2$, respectively. We say that $L$ is \emph{well-ordered} if so is $(L,<_\ell)$.
For example, the previous three languages are well-ordered but $a^*b$ is not (as it contains an infinite descending chain
$\ldots<aab<ab<b$).

When $K$ and $L$ are languages,
then their product is $K\cdot L=\{uv:u\in K,v\in L\}$ and if $u\in\Sigma^*$, then the \emph{left quotient of $L$ with respect to $u$}
is $u^{-1}L=\{v\in\Sigma^*:uv\in L\}$, and of course, $K\cup L=\{u:u\in K\hbox{ or }u\in L\}$ is their \emph{union}.
We write $K<_\ell L$ if $u<_\ell v$ for each $u\in K$ and $v\in L$. Thus, if $K<_\ell L$, then viewing them as the linear orderings
$(K,<_\ell)$ and $(L,<_\ell)$ we get their \emph{sum} $K+L~=~(K\cup L,<_\ell)$. We put an emphasis here on the fact that
taking the sum of two languages $K$ and $L$ is a \emph{partial} operation, defined only if $K<_\ell L$.

When $L$ is a language and $u$ is a (possibly infinite) word, then let $L^{<u}$ and $L^{\geq u}$ respectively denote the languages
$\{v\in L:v<u\}$ and $\{v\in L:v\geq u\}$. Then clearly, $L=L^{<u}+L^{\geq u}$ for any $L$ and $u$.
Note that $L^{<\varepsilon}=\emptyset$ and $L^{\geq\varepsilon}=L$, and also $L^{<a\cdot u}=L^{<a}+ a\bigl((a^{-1}L)^{<u}\bigr)$,
$L^{\geq a\cdot u}=a\bigl((a^{-1}L)^{\geq u}\bigr) + L^{\geq b}$ for the least letter $b$ with $a<b$, if such a letter exists
and $L^{\geq a\cdot u}=a\bigl((a^{-1}L)^{\geq u}\bigr)$ if $a$ is the last letter of the alphabet.
Moreover, $(K\cup L)^{>u}=K^{>u}\cup L^{>u}$ and $(K\cup L)^{\geq u}=K^{\geq u}\cup L^{\geq u}$.

A \emph{context-free grammar} is a tuple $G=(N,\Sigma,P,S)$ with $N$ and $\Sigma$ being the disjoint alphabets of nonterminal and terminal
symbols respectively, $S\in N$ is the \emph{start symbol} and $P$ is a finite set of \emph{productions} of the form $A\to \balpha$
with $A\in N$ being a nonterminal and $\balpha$ being a \emph{sentential form}, i.e. $\balpha=X_1\ldots X_n$ for some 
$n\geq 0$ and $X_1,\ldots,X_n\in N\cup\Sigma$. If $\balpha=uX\bbeta$ for some $u\in\Sigma^*$, $X\in N$ and $\bbeta\in(N\cup\Sigma)^*$,
and $X\to\boldsymbol{\gamma}$ is a production, then $\balpha$ can be rewritten to $u\boldsymbol{\gamma}\bbeta$, which is denoted by $\balpha\Rightarrow u\boldsymbol{\gamma}\bbeta$.
The reflexive-transitive closure of the relation $\Rightarrow$ is denoted by $\Rightarrow^*$.
For any set $\Delta$ of sentential forms, the \emph{language generated by $\Delta$} is
$L(\Delta)~=~\{u\in\Sigma^*:\balpha\Rightarrow^*u\hbox{ for some }\balpha\in\Delta\}$.
For brevity, when $\Delta=\{\balpha_1,\ldots,\balpha_n\}$ is finite, we simply write $L(\balpha_1,\ldots,\balpha_n)$.
Moreover, $o(\Delta)$ denotes $o(L(\Delta))$.

The \emph{language $L(G)$ generated by $G$} is $L(S)$. Languages generated by context-free grammars are called context-free languages.
Two context-free grammars $G$ and $G'$ over the the same terminal alphabet are
\emph{equivalent} if $L(G)=L(G')$ and \emph{order-equivalent} if $o(L(G))=o(L(G'))$.
Any context-free grammar generating a nonempty language of nonempty words can be
effectively transformed into a \emph{Greibach normal form} in which the following all hold:
\begin{itemize}
	\item each production has the form $X\to aX_1\ldots X_n$ for some $a\in\Sigma$,
	\item 	each nonterminal $X$ is \emph{productive}, i.e., $L(X)\neq\emptyset$,
			and \emph{accessible}, i.e., $S\Rightarrow^*uX\balpha$ for some $u\in\Sigma^*$ and $\balpha\in(N\cup\Sigma)^*$.
\end{itemize}
Also, considering the grammar $G'=(N\cup\{S'\},\Sigma,P\cup\{S'\to aS\},S')$ for a fresh symbol $S'$, we get $L(G')=a\cdot L(G)$,
and the order type of each $X\in N$ is the same in both cases, and of course $o(S)=o(S')$.
Thus, to each grammar $G$ one can effectively construct another one $G'$ in Greibach normal form, with $o(L(G))=o(L(G'))$.

Suppose $\balpha=aX_1\ldots X_n$ is a sentential form of a context-free grammar $G=(N,\Sigma,P,S)$ in Greibach normal form
and $b$ is a terminal symbol. Then we define $\balpha^{<b}$, $\balpha^{\geq b}$ and $b^{-1}\balpha$ as the following finite sets of sentential forms:
\begin{align*}
\balpha^{<b}&=\begin{cases}
\{\balpha\}&\hbox{if }a<b\\
\emptyset&\hbox{otherwise}
\end{cases}
&
\balpha^{\geq b}&=\begin{cases}
\emptyset&\hbox{if }a<b\\
\{\balpha\}&\hbox{otherwise}
\end{cases}
\end{align*}

\begin{align*}
b^{-1}\balpha&=\begin{cases}
\{\varepsilon\}&\hbox{if }a=b\hbox{ and }n=0\\
\{X_1\ldots X_n\}&\hbox{if }a=b,~n>0\hbox{ and }X_1\in\Sigma\\
\{\boldsymbol{\delta} X_2\ldots X_n:X_1\to \boldsymbol{\delta}\in P\}&\hbox{if }a=b,~n>0\hbox{ and }X_1\in N\\
\emptyset&\hbox{otherwise}
\end{cases}
\end{align*}
Then clearly, $L(\balpha^{<b})=L(\balpha)^{<b}$, $L(\balpha^{\geq b})=L(\balpha)^{\geq b}$ and $L(b^{-1}\balpha)=b^{-1}L(\balpha)$.
Extending these definitions with $\varepsilon^{<b}=\{\varepsilon\}$, $\varepsilon^{\geq b}=b^{-1}\varepsilon=\emptyset$
and the recursion
\begin{align*}
\balpha^{<b\cdot u}&=\balpha^{<b}\cup \{b\cdot(\boldsymbol{\gamma}^{<u}):\boldsymbol{\gamma}\in b^{-1}\balpha\}&
\alpha^{\geq b\cdot u}&=\{b\cdot(\boldsymbol{\gamma}^{\geq u}):\boldsymbol{\gamma}\in b^{-1}\balpha\}\cup\balpha^{\geq c}
\end{align*}
where $c\in\Sigma$ is the first letter with $b<c$ if such a $c$ exists, otherwise
\begin{align*}
\balpha^{\geq b\cdot u}&=\{b\cdot(\boldsymbol{\gamma}^{\geq u}):\boldsymbol{\gamma}\in b^{-1}\balpha\}
\end{align*}
and $(a\cdot u)^{-1}\balpha=\bigcup\bigl(u^{-1}\boldsymbol{\gamma}:\boldsymbol{\gamma}\in a^{-1}\balpha\bigr)$ we have $L(\balpha^{<u})=L(\balpha)^{<u}$,
$L(\balpha^{\geq u})=L(\balpha)^{\geq u}$ and $L(u^{-1}\balpha)=u^{-1}L(\balpha)$ for any sentential form $\balpha$ not beginning
with a nonterminal and word $u$, moreover, each member of any of these sets is still a sentential form not beginning
with a nonterminal. Clearly, $\balpha^{<u}$, $\balpha^{\geq u}$ and $u^{-1}\balpha$ are all computable for any $u$ and $\balpha$.

A context-free grammar $G=(N,\Sigma,P,S)$ is called an \emph{ordinal grammar} if $o(X)$ is an ordinal
and $L(X)$ is a \emph{prefix-free language} (that is, there are no words $u,v\in L(X)$ with $u<_pv$) for each nonterminal $X\in N$.
It is known \cite{DBLP:journals/fuin/BloomE10} that to each well-ordered context-free language $L$ there exists an ordinal grammar $G$ generating $L$.
It is also known that for \emph{regular} grammars (in which each production has the form $A\to uB$ or $A\to v$)
generating a well-ordered language $L$, order equivalence is decidable~\cite{DBLP:journals/ita/Thomas86},
while for general context-free grammars, it is undecidable whether $o(L(G))=o(L(G'))$ for two grammars $G$ and $G'$:
it is already undecidable whether $o(L(G))=\eta$ holds (or that whether $o(L(G))$ is dense)~\cite{ESIK2011107}.
In contrast, it is decidable whether $L(G)$ is well-ordered~\cite{10.1007/978-3-642-22321-1_19}.

It is unknown whether the order-equivalence problem is decidable for two grammars generating well-ordered languages.

In this paper we show that it is decidable whether $o(L(G))=o(L(G'))$ for two \emph{ordinal} grammars $G$ and $G'$.
Thus, if there is an algorithm that constructs an ordinal grammar $G'$ for an input context-free grammar $G$
generating a well-ordered language (it is known that such an ordinal grammar $G'$ exists but the proof is nonconstructive),
then the order-equivalence problem is decidable for well-ordered context-free languages.
As any finite system $E$ of fixed point equations over variables taking ordinals as values
can effectively by transformed into an ordinal grammar $G$ such that $o(L(G))$
coincides with the least fixed point of the first component of $E$~\cite{DBLP:journals/fuin/BloomE10},
we also get as a byproduct that 
the Cantor normal form of an algebraic ordinal,
given by a finite system of fixed point equations, is
effectively computable. Thus, the isomorphism problem of algebraic ordinals is decidable.

\section{Ordinal grammars}
In this section we recall some known properties of ordinal grammars
and then we prove that the order type of the lexicographic ordering of a language, given by an ordinal grammar, is computable.

It is known from~\cite{DBLP:journals/fuin/BloomE10,10.1007/978-3-642-29344-3_25} that the following are equivalent for an ordinal $\alpha$:
\begin{enumerate}
	\item $\alpha<\omega^{\omega^\omega}$.
	\item $\alpha=o(L(G))$ for a context-free grammar $G$.
	\item $\alpha=o(L)$ for a deterministic context-free language $L$.
	\item $\alpha=o(L(G))$ for an ordinal grammar $G$.
\end{enumerate}

If $G=(N,\Sigma,P,S)$ is a context-free grammar, we define the relation $\preceq$ on $N\cup\Sigma$ as follows: $Y\preceq X$ if and only if $X\Rightarrow^*\balpha Y\bbeta$
for some $\balpha,\bbeta\in(N\cup\Sigma)^*$. Clearly, $\preceq$ is reflexive and transitive (a preorder): $X\approx Y$ denotes that $X\preceq Y$ and $Y\preceq X$ holds.
An equivalence class of $\approx$ is called a \emph{component} of $G$.
If $Y\preceq X$ and they do not belong to the same component, we write $Y\prec X$. As an extension, when $\balpha=X_1\ldots X_n$
is a sentential form with $X_i\prec X$ for each $i\in[n]$, we write $\balpha\prec X$.
Productions of the form $X\to\balpha$ with $\balpha\prec X$ are called \emph{escaping} productions, the others
(when $X_i\approx X$ for some $i\in[n]$) are called \emph{component} productions.

A nonterminal $X$ is called \emph{recursive} if $X\Rightarrow^+\balpha X\bbeta$ for some $\balpha,\bbeta\in(N\cup\Sigma)^*$.

The following are known for ordinal grammars having only usable nonterminals:
\begin{lemma}[\cite{DBLP:journals/fuin/BloomE10}, Proposition 4.9]
\label{lem-product}
If $G$ is an ordinal grammar, then for any word $X_1\ldots X_n\in(\Sigma\cup N)^*$, $o(X_1\ldots X_n)=o(X_n)\times o(X_{n-1})\times\ldots\times o(X_1)$.
\end{lemma}
We will frequently use the above Lemma in the following form: if $X\to X_1\ldots X_n$ is a production of the ordinal grammar $G$
(and thus $L(X_1\ldots X_n)\subseteq L(X)$), then
$o(X_n)\times o(X_{n-1})\times\ldots\times o(X_1)\leq o(X)$.
\begin{lemma}[\cite{DBLP:journals/fuin/BloomE10}, Propositions 4.11, 4.15 and 4.16]
	\label{lem-u0}
	To each recursive nonterminal $X$ there exists a nonempty word $u_X$ such that
	if $X\Rightarrow^+ uX\balpha$ for some $u\in\Sigma^*$ and $\balpha\in(N\cup\Sigma)^*$,
	then $u\in u_X^+$.
	
	Moreover, whenever $X\Rightarrow^*w$ for some word $w$, then $w<_s u_X^\omega$.
\end{lemma}
\begin{lemma}[\cite{DBLP:journals/fuin/BloomE10}, Corollary 4.10]
	\label{lem-monotone}
	If $Y\preceq X$ for the symbols $X,Y\in N\cup\Sigma$, then $o(Y)\leq o(X)$. So if $X\approx Y$, then $o(X)=o(Y)$.
\end{lemma}

For the rest of the section, let $G=(N,\Sigma,P,S)$ be an ordinal grammar.
Since it is decidable whether $L(G)$ is finite, and in that case its order type $o(G)=|L(G)|$
is computable, we assume from now on that $L(G)$ is infinite.

Without loss of generality we can assume that $G$ is in \emph{normal form}:
\begin{itemize}
	\item $G$ has only usable nonterminals: for each $X$, there are words $u,v,w\in\Sigma^*$ with $S\Rightarrow^*uXv$ and
	$X\Rightarrow^*w$.
	\item $L(X)$ is infinite for each nonterminal $X$;
	\item Each production in $P$ has the form $A\to a\balpha$ for some $A\in N$, $a\in\Sigma$ and $\balpha\in(N\cup\Sigma)^*$;
	\item All nonterminals different from $S$ are recursive.
\end{itemize}
To see that such a normal form is computable, consider the following sequence of transformations, starting from an ordinal grammar $G$:
\begin{enumerate}
	\item Unusable nonterminals are eliminated applying the usual algorithm~\cite{Hopcroft+Ullman/79/Introduction}.
	\item If $L(A)$ is finite for some nonterminal $A$, then $A$ gets replaced by all the members of $L(A)$
	on each right-hand side and gets erased from the set of nonterminals. The result of this transformation is still an ordinal grammar.
	\item In particular, if $A\Rightarrow^*\varepsilon$, then by prefix-freeness of $L(A)$ we get that $L(A)=\{\varepsilon\}$, so after this step no $\varepsilon$-transitions remain.
	\item Chain rules of the form $A\to B$ with $A,B\in N$ also get eliminated by the usual algorithm which still
	outputs an ordinal grammar as the generated languages do not change.
	\item By Lemma~\ref{lem-u0}, there are no left-recursive nonterminals, that is, no $A\in N$ with $A\Rightarrow^+A\balpha$ for some $\balpha\in(N\cup\Sigma)^*$.
	Hence, the relation $B<A$ if $A\Rightarrow^+B\balpha$ for some $\balpha\in(N\cup\Sigma)^*$ is a partial ordering. Thus, if we replace each rule
	of the form $A\to B\balpha$ by $A~\to~\bbeta_1\balpha~|~\bbeta_2\balpha~|~\ldots~|~\bbeta_k\balpha$ where $\bbeta_1,\ldots,\bbeta_k$ are all the alternatives of $B$,
	the process eventually terminates.
	\item Finally, if $X\neq S$ is a nonrecursive nonterminal with $X~\to~\balpha_1~|~\ldots~|~\balpha_n$ being all the alternatives of $X$,
	let us erase $X$ from $N$ and replace $X$ by one of the $\balpha_i$'s in all possible ways
	in the right-hand sides of the productions.
	Clearly, this transformation does not change $L(Y)$ for any $X\neq Y$ and reduces the number of nonterminals in $G$.
	Applying this transformation for each nonrecursive nonterminal different from $S$ in some arbitrary order now
	results in an ordinal grammar in normal form.
\end{enumerate}
Clearly, for each $X$ it is decidable whether it is recursive,
and if so, then an $u\in\Sigma^+$ can be computed for which $X\Rightarrow^+uX\balpha$ for some $\balpha\in(N\cup\Sigma)^*$.
Thus, $u_X$ can be chosen as the (still computable) primitive root~\cite{shyr1991free} of $u$.

We can show also the following:
\begin{lemma}
\label{lem-normal-form}
If $G=(N,\Sigma,P,S)$ is an ordinal grammar in normal form, then for each rule $X\to X_1\ldots X_n$ in $P$ one of the following holds:
\begin{enumerate}
	\item either the production is an escaping one (clearly, for a nonrecursive nonterminal this is the only option),
	\item or $X_i\approx X$ for a unique index $i\in[n]$, and $X_j\in\Sigma$ for each $j<i$.
\end{enumerate}
\end{lemma}
\begin{proof}
	Assume that there is a production $X\to X_1\ldots X_n$ for which none of the conditions hold.
	This can happen in the following two cases:
	\begin{enumerate}
		\item If there are at least two indices $i<j$ with $X_i\approx X_j\approx X$, then by Lemma~\ref{lem-product} we get $\alpha\times o(X)\times \beta\times o(X)\times \gamma\leq o(X)$
		for some nonzero ordinals $\alpha,\beta$ and $\gamma$, which is nonsense since if $G$ is in normal form, $L(X)$ is infinite, thus $o(X)>1$.
		\item Similarly, assume there is a unique index $i\in[n]$ with $X_i\approx X$ (thus, $X_j\prec X$ for each $j\neq i$) and $X_j$ is a nonterminal for some $j<i$.
		Then again by Lemma~\ref{lem-product} we get $\alpha\times o(X_i)\times\beta\times o(X_j)\times\gamma\leq o(X)=o(X_i)$ for some nonzero ordinals $\alpha,\beta$ and $\gamma$.
		Since with $X_j$ being a nonterminal we have $o(X_j)>1$, this is again a contradiction.
	\end{enumerate}
\end{proof}

\subsection{Operations on languages}
In this subsection we aim to show that whenever $\balpha\in(N\cup\Sigma)^*$ for some ordinal grammar $G=(N,\Sigma,P,S)$ in normal form,
both the supremum $\bigvee L(\balpha)$ and whether
$\bigvee L(\balpha)$ is a member of $L(\balpha)$ or not,
are computable and also a technical decidability lemma which will be used in the proof of
Theorem~\ref{thm-computable-recursive-themingeszishere}.

Let $X$ be a recursive nonterminal.
By Lemma~\ref{lem-u0}, for each $X\Rightarrow^+w$ we have $w<_s u_X^\omega$, so $u_X^\omega$ is an upper bound
of $L(X)$.
It is also clear that if $X\Rightarrow^+u_X^tXv$, then $X\Rightarrow^+u_X^{t\cdot k}Xv^k$ for every $k\geq 0$.
Hence for any integer $N>0$ there is a word $w\in L(X)$ (say, $w=u_X^{N\cdot t}w'v^N$ where $w'\in L(X)$ is an arbitrary
fixed word) such that $u_X^N<_\ell w$, and as $\mathop\bigvee\limits_{N\geq 0}u_X^N=u_X^\omega$, we immediately
get:
\begin{lemma}
\label{lem-recursive-supremum}
	Suppose $X$ is a recursive nonterminal. Then $\mathop\bigvee L(X)=u_X^\omega$.
	(Thus in particular, there is no largest element in $L(X)$, since $L(X)$ consists of finite words only.)
\end{lemma}

It is obvious that for any $a\in\Sigma$ we have $\bigvee L(a)=a$ and $a\in L(a)$.
For the case of nonrecursive nonterminals (that can be at most $S$)
we need to handle the operations union and product. For union,
we of course have $\bigvee (K\cup L)~=~\bigvee K\vee\bigvee L$ and this element $u$ belongs to $K\cup L$ if
and only if $u=\bigvee K$ and $u\in K$, or $u=\bigvee L$ and $u\in L$ holds.

For product, we state a useful property first:
\begin{proposition}
\label{prop-strict}
If $L$ is prefix-free and $\bigvee L$ exists, then either $L<_s\bigvee L$, or $\bigvee L\in L$ holds.
\end{proposition}
\begin{proof}
Assume neither of the two cases hold for the supremum of $L$. Then, since $\bigvee L\notin L$, we have
$L<_\ell \bigvee L$. Thus, since $L\nless_s\bigvee L$, there is a word $u\in L$ with $u\nless_s\bigvee L$
and $u<_\ell \bigvee L$, hence $u<_p\bigvee L$. But since $L$ is prefix-free, there is no word $v\in L$
with $u<_pv$, thus -- as there is no largest element in $L$ by $\bigvee L\notin L$ -- there is a word
$v\in L$ with $u<_s v$. But as $u<_p\bigvee L$, this yields $\bigvee L<_s v$, a contradiction since
$v<_\ell\bigvee L$ has to hold.
\end{proof}
This proposition entails the following:
\begin{corollary}
\label{cor-strict}
If $K$ and $L$ are nonempty prefix-free languages and both $\bigvee L$ and $\bigvee K$ exist, then
\begin{align*}
\bigvee (KL) &= \begin{cases}
\bigvee K&\hbox{ if }K<_s\bigvee K;\\
\bigvee K\cdot \bigvee L&\hbox{ otherwise},
\end{cases}
\end{align*}
and $\bigvee(KL)\in KL$ if and only if $K\in \bigvee K$ and $L\in\bigvee L$.
\end{corollary}
\begin{proof}
If $K<_s\bigvee K$, then $K\Sigma^*<_s\bigvee K$, so $\bigvee K$ is an upper bound of $KL$ in that case. To see
it's the smallest one, assume $u<_\ell\bigvee K$. Since $\bigvee K$ is the supremum of $K$ with respect to
the total ordering $<_\ell$, this means $u<_\ell v$ for some $v\in K$. But for this $v$ and an arbitrary $w\in L$
we still have $u<_\ell vw$, hence $u$ cannot be an upper bound of $KL$. Thus, $\bigvee K=\bigvee(KL)$.

If $u=\bigvee K\in K$, then for any word $v\in K$ and $w\in L$ we have either $v<_s u$, in which case $vw<_s ux$
for any word $x\in\Sigma^{\leq\omega}$, or $v=u$, in which case $vw\leq_\ell u\bigvee L$ since $w\leq_\ell \bigvee L$.
Thus, $\bigvee K\cdot\bigvee L$ is an upper bound of $\bigvee(KL)$.
Again, if $v<_\ell u\bigvee L$ for some $v$,
then either $v<_\ell u$ in which case $v<_\ell uw\in KL$ for any $w\in L$, thus
$v$ cannot be the supremum of $KL$, or $u<_pv$ in which case $v=uw$ for some $w$ with
$w<_\ell\bigvee L$.
This in turn implies the existence of some $w'\in L$ with $w<_\ell w'$, thus $v=uw<_\ell uw'\in KL$, hence $v$ cannot be an upper bound of $KL$,
showing the claim.

The statement on membership is clear.
\end{proof}

\begin{corollary}
\label{cor-sup-computable}
For any ordinal grammar $G=(N,\Sigma,P,S)$ in normal form and $\balpha\in(N\cup\Sigma)^*$,
the supremum $\bigvee L(\balpha)$ is computable and one of the following cases holds:
\begin{itemize}
	\item $\bigvee L(\balpha)=u$ for some finite $u\in\Sigma^*$, and $u\in L(\balpha)$;
	\item $\bigvee L(\balpha)=uv^\omega$ for some finite $u\in\Sigma^*$ and $v\in\Sigma^+$, and 
	      (of course) $uv^\omega\notin L(\balpha)$.
\end{itemize}
\end{corollary}
\begin{proof}
	We already established $\bigvee L(X)=u_X^\omega$ when $X$ is a recursive nonterminal and that
	$\bigvee L(a)=a\in L(a)$ for terminals $a\in \Sigma$.
	
	Also, for any $\balpha=X_1X_2\ldots X_n\in (N\cup\Sigma)^+$ we can compute $\bigvee L(\balpha)$
	with the recursion
	\begin{align*}
	\bigvee L(X_1\ldots X_n) &=\begin{cases}
	\varepsilon&\hbox{if }n=0\\
	\bigvee(X_1)&\hbox{if }n>0\hbox{ and }\bigvee X_1=uv^\omega\hbox{ for some }u\in\Sigma^*,v\in\Sigma^+\\
	u\cdot\bigvee L(X_2\ldots X_n)&\hbox{if }n>0\hbox{ and }\bigvee X_1=u\in\Sigma^*
	\end{cases}
	\end{align*}
	using Corollary~\ref{cor-strict}.
	
	Then, if $X=S$ is a nonrecursive nonterminal and $X~\to~\balpha_1~|~\balpha_2~|\ldots|~\balpha_n$
	are all the alternatives for $X$, then we have $\bigvee L(X)=\bigvee\limits_{i=1}^n L(\balpha_i)$,
	which yields an inductive proof for the only possible nonrecursive nonterminal $S$.
\end{proof}
Concluding the subsection, we show the following technical lemma:
\begin{lemma}
	\label{lem-transducer}
	It is decidable for any context-free language $L\subseteq\Sigma^*$ and words $u,v$, whether there exists an integer
	$N\geq 0$ such that $uv^N\Sigma^*\cap L~=~\emptyset$. (If so, then $uv^M\Sigma^*\cap L=\emptyset$ for each $M\geq N$.)
\end{lemma}
\begin{proof}
	Let us define the following generalized sequential mappings $f,g:\Sigma^*\to a^*$: let
	\begin{align*}
	f(x)&=\begin{cases}
	g(y)&\hbox{if }x=uy\\
	\varepsilon&\hbox{otherwise,}
	\end{cases}
	&
	g(x)&=\begin{cases}
	a\cdot g(y)&\hbox{if }x=vy\\
	\varepsilon&\hbox{otherwise.}
	\end{cases}
	\end{align*}
	We have that if $uv^N\Sigma^*\cap L$ is nonempty, then $f(L)$ contains some word of length at least $N$,
	and also, if $a^N\in f(L)$, then $uv^N\Sigma^*\cap L$ is nonempty.
	Thus, there is such an integer $N$ satisfying the condition of the lemma if and only if $f(L)$ is finite,
	which is decidable, since the class of context-free languages is effectively closed under generalized
	sequential mappings~\cite{Ginsburg:1966:MTC:1102023}.
\end{proof}		

\subsection{The order type of recursive nonterminals}

In this subsection we show that $o(X)$ is computable, whenever $X$ is a recursive nonterminal of an ordinal grammar $G=(N,\Sigma,P,S)$.

Clearly, for each $a\in\Sigma$ we have $o(L(a))=1$. We will apply induction on the \emph{height} of $X$, defined as the length of the longest
chain $X_1\prec X_2\prec\ldots\prec X_n=X$ with each $X_i$ in $N\cup\Sigma$. (Thus, the height of the terminals is $0$, nonterminals have
positive height.)

Since $X$ is a recursive nonterminal, by Lemma~\ref{lem-u0} there is a (shortest, computable) nonempty word $u_X$ such that
\begin{enumerate}
	\item $w<_su_X^\omega$ for each $w\in L(X)$;
	\item whenever $X\Rightarrow^+uX\balpha$ for some $u\in\Sigma^*$ and $\balpha\in(N\cup\Sigma)^*$, then $u\in u_X^+$.	
\end{enumerate}
This also implies that whenever $X$ and $Y$ are nonterminals belonging to the same component, then there is a unique word $u_{(X,Y)}<_pu_X$ such that $u_X^\omega = u_{(X,Y)}u_Y^\omega$.
Moreover we have:
\begin{proposition}
	\label{prop-betakisebb}
	If $Y\to\bbeta$ is an escaping production for $X\approx Y$, then $u_{(X,Y)}\cdot L(\bbeta)<_su_X^\omega$.
\end{proposition}
\begin{proof}
	In this case, $X\Rightarrow^+u_{(X,Y)}Y\balpha$ for some sentential form $\balpha$. Since $L(\bbeta)\subseteq L(Y)<_s u_Y^\omega$,
	we get $u_{(X,Y)}\cdot L(\bbeta)<_s u_{(X,Y)}u_Y^\omega=u_X^\omega$.
\end{proof}

Now by Lemma~\ref{lem-normal-form} we can deduce that any (leftmost) derivation from $X$ has the form
\begin{align}
\label{eq-leftmost}
X&\Rightarrow~u_1X_1\balpha_1~\Rightarrow~u_1u_2X_2\balpha_2\balpha_1~\Rightarrow~\ldots\\&\Rightarrow~u_1u_2\ldots u_nX_n\balpha_n\ldots\balpha_2\balpha_1~\Rightarrow~u_1u_2\ldots u_n\bbeta\balpha_n\ldots\balpha_2\balpha_1
~\Rightarrow^*~w\nonumber
\end{align}
for some integer $n\geq 0$, nonempty words $u_1,\ldots,u_n\in\Sigma^+$ with $u_1\ldots u_n<_pu_X^\omega$, sentential forms $\balpha_1,\ldots,\balpha_n, \bbeta \in(N\cup\Sigma)^*$ with
$\bbeta\prec X$, $X_i\approx X$ and $\balpha_i\prec X$ for each $i\in[n]$.

By induction, $o(\bbeta)$ is computable
(applying Lemma~\ref{lem-product}) for each possible $\bbeta\prec X_i$ with $X_i\approx X$ and production $X_i\to \bbeta$.
Moreover, $o(\balpha)$ is also computable for each $\balpha\prec X_i$ with a production $X_i\to u_iX_{i+1}\balpha$, $X_{i+1}\approx X_i$
as there are only finitely many such productions.

Let $v_1<_s v_2<_s\ldots<_sv_\ell$ be the complete enumeration of those words $v_i$ with $v_i<_s u_X$ having the form $v_i=ua$ with $u<_pu_X$.

Observe that $L=L(X)$ is the disjoint union of languages of the form $u_X^Nv_i\Sigma^*~\cap~L(X)$, with $N\geq 0$ and $1\leq i\leq \ell$.
Moreover, whenever $u\in u_X^Nv_i\Sigma^*$ and $v\in u_X^Mv_j\Sigma^*$, then $N<M$ or ($N=M$ and $i<j$) implies $u<_sv$. Thus,
these languages form an $\omega$-sequence with respect to the lexicographic ordering and we can write $L$ as
\[L~=~L_1+L_2+L_3+\ldots\]
We will construct an increasing sequence of ordinals
\[o_1~\leq~o_2~\leq~o_3~\leq~\ldots\]
such that the following hold:
\begin{itemize}
	\item for each $i\geq 1$, there is a $j\geq 1$ with $o(L_i)\leq o_j$ and
	\item for each $j\geq 1$, there is an $i\geq 1$ with $o_j\leq o(L_i)$.
\end{itemize}
This implies $o(L)~=~o_1~+~o_2~+~o_3~+~\ldots$. Indeed: by the first condition we have
\[o(L) ~=~o(L_1)~+~o(L_2)~+~\ldots\\
~\leq~o_{f(1)}~+~o_{f(2)}~+~\ldots
\]
for some indices $f(1)$, $f(2)$ and so on. Let us define for each $j$ the index $g(j)$ as follows: $g(1)=f(1)$
and for each $j>1$, let $g(j)=\max\{g(j-1)+1,f(j)\}$. Then we have
$o(L)\leq o_{g(1)}+o_{g(2)}+\ldots$ and $g(1)<g(2)<\ldots$. Thus, $o(L)\leq o_1+o_2+\ldots$ holds (as the former order type is a sub-order type of the latter), 
the other direction being symmetric.

Let us now consider one such language $L_t$. Then, $L_t$ is a finite union of languages of the form
\begin{align}
\label{eqn-theyhavethisform}
u_1u_2\ldots u_nL'L(\balpha_n)\ldots L(\balpha_2)L(\balpha_1)
\end{align}
where $u_1\ldots u_n<_p (u_X)^N$ for some $N$ depending only on $t$, moreover, applying Proposition~\ref{prop-betakisebb}
we get that each such $L'$
has the form $\bigl((u_1\ldots u_n)^{-1}u_X^Nv_j\bigr)\Sigma^*~\cap~L(\bbeta)~=~u_{X'}^{M}v\Sigma^*~\cap~L(\bbeta)$, and for each $i\geq 0$
there is a production of the form $X_i\to u_iX_{i+1}\balpha_i$ (recall that due to the normal form
each $u_i$ is nonempty) for some nonterminals $X_i\approx X$,
$X_1=X$ and $X_{n+1}\to\bbeta$ with $\bbeta\prec X$. Clearly, for any fixed $N$ and $v_i$, there are only finitely many such choices.

We do not have to explicitly compute the order type of each such $L_t$, due to the following lemma:
\begin{lemma}
\label{lem-uniodeg}
	Assume $o_1\leq o_2\leq \ldots$ is a sequence of ordinals and $K$, $L$ are languages with
	$\deg(o(L))$, $\deg(o(K))<\deg\bigl(\bigvee o_i\bigr)$. Then $o(K\cup L)<o_j$ for some index $j$.
\end{lemma}
\begin{proof}
	Without loss of generality, let $o(L)\leq o(K)$.	
	By Theorem~\ref{thm-union} we have that $o(K\cup L)<\omega^{\deg(o(K))}\times T$ for some integer $T$.	
	It suffices to show that for each integer $T>0$, there exists an $o_i$ with 
	$o_i> \omega^{\deg(o(K))}\times T$. Assume to the contrary that each $o_i$ is at most $\omega^{\deg(o(K))}\times T'$
	for some integer $T'$. But then, $\bigvee o_i\leq \omega^{\deg(o(K))}\times T'$ and thus  $\deg(\bigvee o_i)\leq \deg(o(K))$,
	a contradiction.
\end{proof}

Equipped by our lemmas we are ready to prove the (technically most involved)
main result of the subsection:
\begin{theorem}
	\label{thm-computable-recursive-themingeszishere}
	Assume $G$ is an ordinal grammar in normal form and $X$ is a recursive nonterminal.
	Let $o_\balpha$ be the maximal order type of some $L(\balpha)$ for which a component production of the form $X'\to uX''\balpha$ exists in $G$ for some $X\approx X'$,
	and $o_\bbeta$ be the maximal order type of some $L(\bbeta)$ with $X'\to \bbeta$ being an escaping production of $G$ with $X'\approx X$.
	
	Then the order type of $L(X)$ is:
	\begin{enumerate}
		\item ${(o_\balpha)}^\omega$ if $o_\bbeta<{( o_\balpha)}^\omega$;
		\item $o_\bbeta$ if $o_\bbeta=\omega^{\deg(o_\bbeta)}$ and for each escaping production $X'\to\bbeta$ with $\deg(o(L(\bbeta)))=\deg(o_\bbeta)$, the
		language $u_{X'}^N\Sigma^*\cap L(\bbeta)$ is nonempty for infinitely many integers $N\geq 0$;
		\item $o_\bbeta\times\omega$, otherwise.
	\end{enumerate}	
\end{theorem}
\begin{proof}
So let $o_\balpha$ be the ordinal $\max\{o(L(\balpha)):~X'\to uX''\balpha\hbox{ is a production for some }X'\approx X''\approx X\}$.
	Since there are only finitely many such $\balpha$, and $\balpha\prec X$ holds for each of them, $o_\balpha$ is well-defined and computable
	by induction.
	
Also, let $o_\bbeta$ be $\max\{o(L(\bbeta)):~X'\to \bbeta\hbox{ is a production for some }X'\approx X,\bbeta\prec X\}$. This ordinal $o_\bbeta$ is well-defined and computable as well.

We also use the shorthands $\gamma=\deg(o_\balpha)$ and $\delta=\deg(o_\bbeta)$.
These ordinals are also computable
(as an ordinal ``being computable'' means in our context that the Cantor normal form of the ordinal is computable).

Now we apply a case analysis, based on $\delta$ and $\gamma$.
We note next to these (sub, subsub)cases to which case of the theorem they correspond.

\subsection*{Case 1: $\delta<\gamma\times\omega$}
This case corresponds to Case $1$ of the theorem.
We claim that in this case $o(X)={(o_\balpha)}^\omega$.
To see this, it suffices to show for each integer $N\geq 0$ that ${(o_\balpha)}^N<o(X)$ and that there is an $L_i$ with $o(L_i)<{(o_\balpha)}^N$.

For ${(o_\balpha)}^N<o(X)$, let $X'\to uX''\balpha$ be a component production with $o(L(\balpha))=o_\balpha$ and
let $u_0,v_0,u_1,v_1\in\Sigma^*$ be so that $X''\Rightarrow^* u_1X'v_1$ and $X\Rightarrow^* u_0X'v_0$.
Finally, let $w\in L(X')$. Then we have \[X\Rightarrow^*u_0(uu_1)^Nw(v_1\balpha)^Nv_0.\]
Since by Lemma~\ref{lem-product} the order type of the language generated by this sentential form is at least ${(o_\balpha)}^N$,
and this language is a subset of $L(X)$, this direction is proved.

For the other direction, note that $\deg({(o_\balpha)}^\omega)=\gamma\times\omega$.
Thus, since each $L_i$ is a finite union of languages of the form~\ref{eqn-theyhavethisform}, in which $L'\subseteq L(\bbeta)$ for some
$\bbeta$, by Lemma~\ref{lem-uniodeg} it suffices to show that \[\deg(o(u_1\ldots u_nL(\bbeta)L(\balpha_n)L(\balpha_{n-1})\ldots L(\balpha_1)))<\gamma\times\omega.\]
But, as each $o(\balpha_i)$ is at most $o_\balpha$ and $o(\bbeta)\leq o_\bbeta$, we get
that this sentential form has the order type at most ${(o_\balpha)}^n\times o_\bbeta$.

We have that $\deg({(o_\balpha)}^n\times o_\bbeta)=\gamma\times n+\delta$ which is smaller than $\gamma\times\omega$ if so is $\delta$ and the claim is proved.

\subsection*{Case 2: $\gamma\times\omega\leq \delta$}
Observe that this case applies if and only if $\deg(\gamma)<\deg(\delta)$
and that this cannot happen within Case $1$ of the theorem.
We split the analysis of this case to several subcases. For each escaping production $X'\to\bbeta$ with $\deg(o(\bbeta))=\delta$,
we decide whether there exists an $N\geq 0$ such that $u_{X'}^N\Sigma^*\cap L(\bbeta)~=~\emptyset$.
By Lemma~\ref{lem-transducer}, this is decidable.

\subsubsection*{Subcase 2.1: $\gamma\times\omega\leq\delta$ and there exists a $\bbeta$ such that $u_{X'}^N\Sigma^*\cap L(\bbeta)~=~\emptyset$ for some $N$}
This subcase rules out Case $2$ of the theorem by the condition $u_{X'}^N\Sigma^*\cap L(\bbeta)~=~\emptyset$, so this subcase falls under Case $3$ of the theorem,
and we claim $o(X)=o_\bbeta\times\omega$ in this subcase.

In this subcase, $L(\bbeta)$ is a finite union of languages
of the form $K_{N,v}~=~u_{X'}^Nv\Sigma^*~\cap~L(\bbeta)$ for some word $v=v'a<_su_{X'}$ with $v'<_p u_{X'}$
(see Figure~\ref{fig-tree}).
Thus, there is one $K_{N,v}$ among these languages with $\deg(o(K_{N,v}))=\delta$ (since the degree of this finite union is $\delta$).
Such a language is a subset of a factor $L'$ of a language of the form~(\ref{eqn-theyhavethisform}), moreover, such an $L'$ occurs as a factor in
infinitely many languages $L_i$: if $X'\Rightarrow^+ u_{X'}^tX'\boldsymbol{\alpha}$, and
$K_{N,v}$ is a subset of one of the languages $L'$ belonging to $L_i$, then it also belongs to the same factor $L'$ of $L_{i+t}$.
Hence, we have the lower bound $\omega^\delta\times\omega=\omega^{\delta+1}={o_\bbeta}\times\omega\leq o(X)$.

To see that this is an upper bound as well, 
it suffices to show that each language of the form~(\ref{eqn-theyhavethisform}) has an order type less than $o_\bbeta\times\omega$, that is, has a degree at most $\delta$.
Again, similarly to Case 1 we get that the order type of such a language is upperbounded by ${(o_\balpha)}^n\times o_\bbeta$ whose degree is $\gamma\times n+\delta$ which is $\delta$ since
the degree of $\gamma$ is smaller than the degree of $\delta$. (In this case it can happen that $o_\balpha<\omega$ but for finite powers, $\deg(\alpha^n)=\deg(\alpha)\times n$
still holds.)

Thus, in this subcase the order type of $L(X)$ is $o_\bbeta\times\omega$.

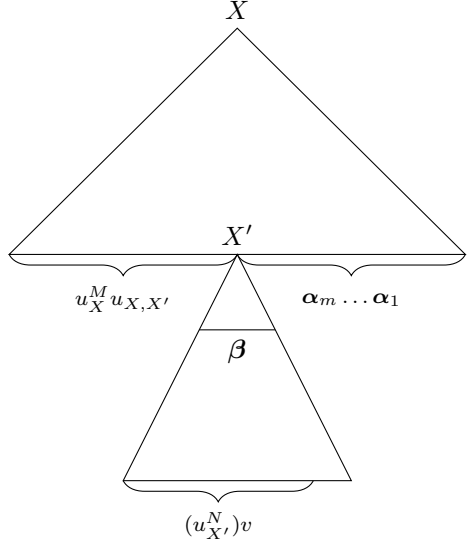
\begin{figure}\centering
	\begin{tikzpicture}
	\draw (0,10) node[anchor=south]{$X$}
	-- (3,7) node[anchor=north]{}
	-- (-3,7) node[anchor=south]{}
	-- cycle;
	\draw [decorate,decoration={brace,amplitude=8pt,mirror,raise=0pt},yshift=0pt]
	(-3,7) -- (0,7) node [black,midway,yshift=-0.6cm] {\footnotesize $u^M_Xu_{X,X'}$};
	\draw [decorate,decoration={brace,amplitude=8pt,mirror,raise=0pt},yshift=0pt]
	(0,7) -- (3,7) node [black,midway,yshift=-0.6cm] {\footnotesize $\balpha_m\ldots\balpha_1$};
	\draw (0,7) node[anchor=south]{$X'$}
	-- (1.5,4) node[anchor=north]{}
	-- (-1.5,4) node[anchor=south]{}
	-- cycle;
	\draw (-0.5,6) -- node[anchor=north]{$\bbeta$} (0.5,6);
	\draw [decorate,decoration={brace,amplitude=8pt,mirror,raise=0pt},yshift=0pt]
	(-1.5,4) -- (1,4) node [black,midway,yshift=-0.6cm] {\footnotesize $(u_{X'}^N)v$};
	
	\end{tikzpicture}
	\caption{The tree depicting the factorization of a word}
	\label{fig-tree}
\end{figure}

\subsubsection*{Subcase 2.2: $\gamma\times\omega\leq\delta$ and for all $\bbeta$ and $N$, $u_{X'}^N\Sigma^*\cap L(\bbeta)~\neq~\emptyset$}
In this subcase, the order type of each such $\bbeta$ can be written as an infinite sum of nonempty ordinals
$o_\bbeta=o_{\bbeta_1}+o_{\bbeta_2}+\ldots$, $L(\bbeta)$ being the ordered disjoint union of the nonempty languages $K_{N,v}$.
Now again, we have two subsubcases: either $o_\bbeta=\omega^\delta$
(this subsubcase corresponds to Case $2$ of the theorem)
or $o_\bbeta>\omega^\delta$ (which in turn falls under Case $3$ of the theorem as well).

{\bf If $o_\bbeta=\omega^\delta$}, then the degree of each such $o_{\bbeta_i}$ is strictly smaller than $\delta$. In this case, each language of the form~(\ref{eqn-theyhavethisform})
has an order type at most ${(o_\balpha)}^n\times o$ for some $o$ with $\deg(o)=\delta'<\delta$,
the degree of which ordinal is $\gamma\times n+\delta'$.
Since $\deg(\gamma\times n)<\deg(\delta)$, we have $\gamma\times n+\delta'<\gamma\times n+\delta=\delta$,
thus each such language $L_i$ has a degree still strictly smaller than $\delta$. Thus, $o_\bbeta=\omega^\delta$ is an upper bound for $o(X)$
in this case. Since $o(\bbeta)$ occurs as a subordering in $o(X)$, we also have $o_\bbeta\leq o(X)$, thus $o(X)=o_\bbeta$ in this subsubcase.

{\bf If $o_\bbeta>\omega^\delta$}, then there exists an $o_{\bbeta_i}$ with degree $\delta$. Proceeding with the argument exactly as in Subcase 2.1, we get that
$o(L)=o_\bbeta\times\omega$ in this subsubcase.
\end{proof}

Thus in particular, as each condition is decidable if the order types $o(\bbeta)$ and $o(\balpha)$ are computable, which are, applying the induction hypothesis,
we get decidability:
\begin{theorem}
\label{thm-recursive}
	Assume $G$ is an ordinal grammar in normal form and $X$ is a recursive nonterminal.
	
	Then $o(X)$ is computable.
\end{theorem}

\subsection{The order type of nonrecursive nonterminals}
Recall that if $G$ is an ordinal grammar in normal form, then its only nonrecursive nonterminal can be its starting symbol $S$. Thus, if $\balpha_1,\ldots,\balpha_n$ are all the alternatives of $S$, then $L(G)=\mathop\bigcup\limits_{i=1}^nL(\balpha_i)$ and all the $\balpha_i$s consist of terminal symbols and recursive
nonterminals, whose order type is already known to be computable.

Hence we only have to show that the following problem is computable:
\begin{itemize}
	\item {\bf Input:} An ordinal grammar $G=(N,\Sigma,P,S)$ (in normal form), and a finite set $\{\balpha_1,\ldots,\balpha_n\}$ of sentential forms
	such that for each symbol $X$ occurring in the set, $o(X)$ is known.
	\item {\bf Output:} The order type of $L=\mathop\bigcup\limits_{i=1}^nL(\balpha_i)$.
\end{itemize}
We claim that the following algorithm $A$ solves this problem:
\begin{lstlisting}[mathescape=true, style=myScalastyle]
function $A(\{\balpha_1,\ldots,\balpha_n\})$
  if( $n$ == $0$ ) return $0$
  $\Right$ := $\{\balpha_1,\ldots,\balpha_n\}$
  $\Left$  := $\emptyset$
  $u$ := $\varepsilon$
  while( true ) {
    $w$ := $\max\{\bigvee L(\balpha) :\balpha\in\mathrm{Right} \}$
    $\Right_1$ := $\{\balpha\in\Right:~\bigvee L(\balpha)<w\}$
    $\Right_2$ := $\{\balpha\in\Right:~\bigvee L(\balpha)=w\}$
    $o$ := $\max\{o(L(\balpha)):\balpha\in\mathrm{Right}_2\}$
    if( $o=\omega^\gamma$ for some $\gamma$ )
      Let $w'$ be a finite prefix of $w$ such that for each $\balpha\in\Right_1$, $L(\balpha)<w'$ already holds.
      return $A(\mathrm{Left})+A\Bigl(\Right_1\cup\bigl\{(\balpha^{<{w'}}):\balpha\in\Right_2\bigr\}\Bigr)~+~\omega^\gamma$
    Let $a$ be the largest letter of $\Sigma$ such that there exists some $a\balpha\in\mathrm{Right}$
    $\mathrm{Left}$ := $\mathrm{Left}\cup\{u\cdot\balpha:\balpha\in\mathrm{Right},~\mathrm{First}(\balpha)\neq a\}$
    $\mathrm{Right}$ := $a^{-1}\mathrm{Right}$
    $u$ := $u\cdot a$
    $\mathrm{Right}$ := $\{\boldsymbol{\delta}\balpha':\exists X\to\boldsymbol{\delta}\in P,X\balpha'\in\mathrm{Right}\}~\cup~\{\balpha:\alpha\in\mathrm{Right},~\mathrm{First}(\balpha)\notin N\}$. 
  }
\end{lstlisting}
In the above algorithm, for a sentential form $\balpha=X\cdot\balpha'$, $\mathrm{First}(\balpha)=X$ and $\mathrm{First}(\varepsilon)=\varepsilon$.

We use induction on $o(L)$ to show that the algorithm always terminates, and it does so with the right answer.
Since $G$ is in normal form, we can restrict the proof to those cases when each $\balpha_i$ is either $\varepsilon$ or starts with a terminal symbol.

If this order type is $0$, then (since each nonterminal is productive as $G$ is in normal form) $n=0$ has to hold,
in which case the algorithm indeed returns $0$. Now assume $o(L)>0$, thus $n>0$.

For the sake of convenience, let $L(\mathrm{Left})$ stand for the language $\bigcup_{\bbeta\in\mathrm{Left}}L(\bbeta)$
and similarly for $L(\Right)$.
We claim that the following invariants are preserved in the loop of the algorithm: 
\begin{align*}
L(\Left)&< u&
&\hbox{and}&
L&=L(\Left)~\cup~u\cdot L(\Right).
\end{align*}
Also, $\Right\neq\emptyset$ and after each execution of Line 7, $u\cdot w=\bigvee L$ .

Upon entering the loop, $\Left=\emptyset$ and from $u=\varepsilon$ we have $u\cdot L(\Right)=L(\Right)=L$.
Within the loop, if $L=L(\Left)\cup u\cdot L(\Right)$ and $L(\Left)< u$ before executing Line 7,
then $\bigvee L~=~\bigvee \bigl(u\cdot L(\Right)\bigr)~=~u\cdot \bigvee L(\Right)~=~u\cdot\max\{\bigvee L(\balpha):\balpha\in\Right\}$,
thus indeed, $u\cdot w=\bigvee L$.

Now assuming $L(\Left)< u$ holds when we start an iteration of the loop, we have to see that
$L(\Left)\cup u\cdot\bigl(\bigcup L(\balpha):\balpha\in\Right,\First(\balpha)\neq a\bigr)<u\cdot a$
for the letter $a$ chosen during Line 14. The part $L(\Left)<u<u\cdot a$ is clear.
The latter part is equivalent to $L(\balpha)<a$ holds for each $\balpha\in\Right$ with $\First(\balpha)\neq a$,
which holds since if such an $\balpha$ begins with a terminal symbol $b$ then by the choice of $a$ we have $b<a$,
and if $\balpha=\varepsilon$, then also $\varepsilon<a$, showing preservation of the property $L(\Left)<u$.
It is also clear that the operation in Line 16 can't make $\Right$ empty by the choice of $a$
(also, since $\Right$ is nonempty and by assumption, each $\balpha\in\Right$ begins with a terminal symbol,
such a letter $a$ always exists: if $\Right=\{\varepsilon\}$, then $o=1=\omega^0$ and the algorithm terminates at Line 13).

Assuming $L=L(\Left)~\cup~u\cdot L(\Right)$ when starting an iteration, after executing Line 17
we have to show that $L=L(\Left)\cup\{u\cdot L(\balpha):\balpha\in\Right,\First(\balpha)\neq a\}~\cup~u\cdot a\cdot L(a^{-1}\Right)$
for the original values of $\Left$ and $\Right$ to see preservance of this property.
But this clearly holds for arbitrary set of sentential forms $\Left$ and $\Right$, thus this property is
again a loop invariant.

After executing Line 16, it may happen that $\Right$ contains some sentential form(s) starting with a nonterminal;
executing Line 18 does not change $L(\Right)$ but restores the property of $\Right$ that each $\balpha\in\Right$
begins with a terminal symbol (or $\balpha=\varepsilon$).

Now by the first two properties we have $o(L)~=~o(L(\Left))+o(u\cdot L(\Right))~=~o(L(\Left))+o(L(\Right))$.

We show that this is exactly the ordinal we return in Line 13, should the condition of Line 11 hold.
Consider the sets $\Right_1$ and $\Right_2$ of sentential forms.
By the definition of $w$, $\Right_2$ is nonempty and $\Right=\Right_1\uplus\Right_2$.
By the choice of $w'$, we have that $L(\Right_1)<w'$ and of course $L(\Right_2)=L({\Right_2}^{<w'})+L({\Right_2}^{\geq w'})$,
thus
\[L~=L(\Left)+\Bigl(L(\Right_1)\cup L({\Right_2}^{<w'})\Bigr)+L({\Right_2}^{\geq w'}).\]
Observe that such a $w'$ is computable as $w$ is a computable word (possibly having the form $xy^\omega$ for some
finite words $x,y$ by Corollary~\ref{cor-sup-computable}), so its prefixes can be enumerated and for each prefix $w_0$, the emptiness of the 
context-free language $L(\balpha^{\geq w_0})$ can be decided for each $\balpha\in\Right_1$; as for these strings $\balpha$
we have $\bigvee L(\balpha)<w$, there is a finite prefix $w_0$ of $w$ with $L(\balpha)$ being already smaller than $w_0$.
Thus, even the shortest such prefix $w'$ of $w$ can be computed.

Since $L({\Right_2}^{\geq w'})$ is nonempty (as $w'<w=\bigvee L(\Right_2)$) we get that $o(L(\Left))$ and
$o\Bigl(L(\Right_1)\cup L({\Right_2}^{<w'})\Bigr)$ are both strictly smaller than $o(L)$,
thus applying the induction hypothesis we get that the algorithm terminates with a correct answer in Line 13
if $o(L({\Right_2}^{\geq w'}))=\omega^\gamma$. Since each nonempty suffix of $\omega^\gamma$ is itself,
and $\omega^\gamma$ is the order type of at least one $L(\balpha)$ with $\balpha\in\Right_2$ by the choice of $o$,
we have $o(\balpha^{\geq w'})=\omega^\gamma$, thus $\omega^\gamma\leq o(L({\Right_2}^{\geq w'}))$.
For the lower bound, note that $L({\Right_2}^{\geq w'})$ is a finite union of languages $L_i$ such that for each $i$,
$\bigvee L_i=w$ and $o(L_i)\leq \omega^\gamma$.
If $\gamma=0$, then all these languages are singletons containing the word $w$ and the claim holds.
Otherwise, none of the languages $L_i$ have a largest element and so
for any word $v\in L({\Right_2}^{\geq w'})$ we have $o({L_i}^{< v})<\omega^\gamma$ 
(by that $v<\bigvee L({\Right_2}^{\geq w'})=\bigvee L_i=w$ and so ${L_i}^{\geq v}$ is nonempty)
and so $L({\Right_2}^{<v})$ is a finite union of languages, each having an order type strictly less than $\omega^\gamma$,
thus the union itself also has an order type less than $\omega^\gamma$. So each prefix of $o(L(\Right_2))$ is
less than $\omega^\gamma$ which makes $o(L(\Right_2))\leq \omega^\gamma$ and the claim holds.

Thus, if the algorithm makes a recursive call in Line 13, then it returns with a correct answer.

We still have to show that the algorithm eventually terminates. To see this, observe that $u\cdot w=\bigvee L$ holds
after each iteration of the loop and $u$ gets longer by one letter in each iteration. Hence, if the algorithm
does not terminate, then the supremum of the values of the variable $u$ is $\bigvee L$. Since $o(L)\neq 0$,
say $o(L)=\omega^{\gamma_1}\times n_1+\ldots +\omega^{\gamma_k}\times n_k$ for some integer $k>0$,
integer coefficients $n_i>0$ and ordinals $\gamma_1>\gamma_2>\ldots>\gamma_k$, so there exists some word
$x\in L$ with $o(L^{\geq x})=\omega^{\gamma_k}$. Clearly, after some finite number (say, $|x|$) of iterations
we have $x<u$, this makes $o(L^{\geq u})=\omega^{\gamma_k}$, and by $u\cdot L(\Right)$ being a nonempty
suffix of $L^{\geq u}$, we get that $o(L(\Right))=\omega^{\gamma_k}$: as $L(\Right_2)\subseteq L(\Right)$
is a finite union of languages, we have $o(L(\balpha))\leq \omega^{\gamma_k}$ for each $\balpha\in\Right_2$
and equality holds for at least one of them. Hence, the loop terminates in at most $|x|$ steps,
finishing the proof of termination as well.

Theorem~\ref{thm-recursive}, in conjunction with the correctness of the above algorithm yields the main result of the paper:
\begin{theorem}
\label{thm-main}
Given an ordinal grammar $G$, one can compute the order type $o(G)$ in Cantor normal form.
\end{theorem}
Applying the construction of~\cite{DBLP:journals/fuin/BloomE10}, we get the following corollary:
\begin{corollary}
\label{cor-main}
The Cantor normal form of an algebraic ordinal, given by a finite system of fixed point equations, is
effectively computable. Thus, the isomorphism problem of algebraic ordinals is decidable.
\end{corollary}

\section{Conclusion and acknowledgement}
We have shown that the isomorphism problem of algebraic ordinals is decidable, studying the order types
of well-ordered context-free languages, given by an ordinal grammar.
It is an interesting question whether the proof can be lifted to scattered linear orders: in many cases,
scattered linear orders behave almost as well as well-orders. Also, it would be interesting to
analyze the runtime of our algorithm: we only know that by well-founded induction the computation
eventually terminates.

The authors wish to thank Prof. Zolt\'an F\"ul\"op for the discussion on generalized sequential
mappings on context-free languages.

\bibliography{biblio}{}
\bibliographystyle{plain}
\end{document}